\documentclass[a4paper,USenglish,cleveref,autoref,thm-restate,nolineno]{socg-lipics-v2019}

\usepackage{tikz}
\usepackage{textgreek} 
\hypersetup{unicode=true} 
\usepackage{csquotes} 

\renewcommand*\Re{\mathbb R}
\newcommand*\Nat{\mathbb N}
\newcommand*\card[1]{\lvert #1\rvert}
\newcommand*\edges{\mathcal{E}}
\DeclareMathOperator*{\VC}{\dim}

\newtheorem*{problem}{Problem}

\title{The \texorpdfstring{$\epsilon$-$t$-Net}{\textepsilon-t-Net} Problem}

\author{Noga Alon}{Department of Mathematics, Princeton University, Princeton, NJ 08544, USA \and Schools of Mathematics and Computer Science, Tel Aviv University, Tel Aviv 69978, Israel}{nogaa@tau.ac.il}{}{Research supported in part by NSF grant DMS-1855464, ISF grant 281/17, GIF grant G-1347-304.6/2016, and the Simons Foundation.}
\author{Bruno Jartoux}{Department of Computer Science, Ben-Gurion University of the Negev, Be'er-Sheva, Israel}{jartoux@post.bgu.ac.il}{https://orcid.org/0000-0002-5341-1968}{Research supported by the European Research Council (ERC) under the European Union’s Horizon 2020 research and innovation programme (Grant agreement No. 678765) and by Grant 635/16 from the Israel Science Foundation.}
\author{Chaya Keller}{Department of Computer Science, Ariel University, Ariel, Israel}{chayak@ariel.ac.il}{}{Part of the research was done when the author was at the Technion, Israel, and was supported by Grant 409/16 from the Israel Science Foundation.}
\author{Shakhar Smorodinsky}{Department of Mathematics, Ben-Gurion University of the Negev, Be'er-Sheva, Israel}{shakhar@math.bgu.ac.il}{https://orcid.org/0000-0003-3038-6955}{Research partially supported by Grant 635/16 from the Israel Science Foundation.}
\author{Yelena Yuditsky}{Department of Mathematics, Ben-Gurion University of the Negev, Be'er-Sheva, Israel}{yuditskyL@gmail.com}{}{Research supported by the European Research Council (ERC) under the European Union’s Horizon 2020 research and innovation programme (Grant agreement No. 678765) and by Grant 635/16 from the Israel Science Foundation.}

\authorrunning{N.\@ Alon, B.\@ Jartoux, C.\@ Keller, S.\@ Smorodinsky, Y.\@ Yuditsky}
\Copyright{Noga Alon, Bruno Jartoux, Chaya Keller, Shakhar Smorodinsky, and Yelena Yuditsky}

\acknowledgements{The authors are grateful to Adi Shamir for fruitful suggestions regarding the application of $\epsilon$-$t$-nets to secret sharing.}

\hideLIPIcs 
\renewcommand\subjclassHeading{\vspace{-1.4\baselineskip}}
\ccsdesc{\smash{}}

\keywords{epsilon-nets, geometric hypergraphs, VC-dimension, linear union complexity}
\relatedversion{An abridged version of this paper is to appear in the proceedings of the 36th international Symposium on Computational Geometry (SoCG 2020).}

\begin{document}
\maketitle

\begin{abstract}
We study a natural generalization of the classical $\epsilon$-net problem (Haussler--Welzl 1987), which we call \emph{the $\epsilon$-$t$-net problem}: Given a hypergraph on $n$ vertices and parameters $t$ and $\epsilon\geq \frac t n$, find a minimum-sized family $S$ of $t$-element subsets of vertices such that each hyperedge of size at least $\epsilon n$ contains a set in $S$. When $t=1$, this corresponds to the $\epsilon$-net problem.

We prove that any sufficiently large hypergraph with VC-dimension $d$ admits an $\epsilon$-$t$-net of size $O(\frac{ (1+\log t)d}{\epsilon} \log \frac{1}{\epsilon})$. 
For some families of geometrically-defined hypergraphs (such as the dual hypergraph of regions with linear union complexity), we prove the existence of $O(\frac{1}{\epsilon})$-sized $\epsilon$-$t$-nets. 

We also present an explicit construction of $\epsilon$-$t$-nets (including $\epsilon$-nets) for hypergraphs with bounded VC-dimension. In comparison to previous constructions for the special case of $\epsilon$-nets (i.e., for $t=1$), it does not rely on advanced derandomization techniques. To this end we introduce a variant of the notion of VC-dimension which is of independent interest.
\end{abstract}

\section{Introduction}

\subsection{Preliminaries}

\subsubsection*{Hypergraphs and VC-dimension.} 

A \textit{hypergraph} is a pair $H=(V,\edges)$ where $V$ is a set of \emph{vertices} and $\edges \subseteq 2^V$ is the set of \textit{hyperedges} of $H$. When $V$ is finite, $H$ is a \textit{finite} hypergraph.

A subset $V'\subseteq V$ is \emph{shattered} if all its subsets are realized by $\edges$, meaning $\{ V'\cap e\colon e\in \edges\} = 2^{V'}$. The \textit{VC-dimension} of $H$, denoted by $\VC H$, is the cardinality of a largest shattered subset of $V$ or $+\infty$ if arbitrarily large subsets are shattered (which does not happen in finite hypergraphs). This parameter plays a central role in statistical learning, computational geometry, and other areas of computer science and combinatorics \cite{VC71,MATOUSEK,MV17}. 

\subsubsection*{$\epsilon$-nets, Mnets.} 

Let $\epsilon\in (0,1)$. An \emph{$\epsilon$-net} for a finite hypergraph $(V,\edges)$ is a subset of vertices $S\subseteq V$ such that $S\cap e\neq \emptyset$ for every hyperedge $e\in \edges$ such that $\card e \geq \epsilon \card V$. 

Haussler and Welzl \cite{HW87} proved that finite hypergraphs with VC-dimension $d$ admit $\epsilon$-nets of size $O(\frac{d}{\epsilon} \log\frac{d}{\epsilon})$, later improved to $O(\frac{d}{\epsilon} \log\frac{1}{\epsilon})$ \cite{KomPaW92}. 
In the last three decades, $\epsilon$-nets have found applications in diverse areas of computer science, including machine learning \cite{BEHW89}, algorithms \cite{Chan18}, computational geometry \cite{AFM12} and social choice \cite{ABKKW}.

Mustafa and Ray introduced the notion of \emph{Mnets} \cite{MR17}. For a hypergraph $(V,\edges)$ and for a fixed $\epsilon\in (0,1)$, an \emph{$\epsilon$-Mnet} is a family $\{V_1,V_2,\ldots,V_{\ell}\}$ such that each $V_i\subseteq V$, each $V_i$ is of size $\Theta(\epsilon \card V)$, and, for each $e\in \edges$ such that $\card e \geq \epsilon \card V$,  $V_i\subseteq e$ for some $V_i$. 
They constructed small $\epsilon$-Mnets (i.e., such families with small $\ell$) for several classes of geometric hypergraphs. These results were extended by Dutta et al.\@ \cite{DGJM19} using {polynomial partitioning}.

\subsubsection*{Explicit constructions} 

Although Hausssler and Welzl's proof of the $\epsilon$-net theorem is probabilistic, several deterministic constructions of $\epsilon$-nets for hypergraphs with finite VC-dimension have been devised \cite{BCM99,Ma95,MC96}. The best result of this kind is Br\"{o}nniman, Chazelle and Matou\v{s}ek's $O(\epsilon^{-d} \log^d \frac{1}{\epsilon}\card{V})$-time algorithm for computing an $\epsilon$-net of size $O(\frac{d}{\epsilon} \log \frac{d}{\epsilon})$ \cite{BCM99}.
These constructions are used to derandomize applications of $\epsilon$-nets, such as low-dimensional linear programming \cite{Chan18}.

In scenarios where the VC-dimension is $\Omega(\log \card V)$, the running time of these constructions becomes exponential in $\card V$. For one such scenario --  the hypergraph induced by half-spaces on the discrete cube $V=\{-1,1\}^d$ -- Rabani and Shpilka \cite{RS08} presented an efficient explicit construction of an $\epsilon$-net, alas of sub-optimal size: $O(\epsilon^{-b}\card{V}^a)$ for some universal constants $a,b>0$, whereas $O(\card V /\epsilon)$ can be obtained by random sampling. Like the aforementioned explicit constructions, the construction of \cite{RS08} is based on derandomization.

\subsection{Our problem}

We denote by $\binom X k$ the set of all subsets of cardinality $k$ (or \enquote{$k$-subsets}) of the set $X$.

\begin{definition}
Let $H=(V,\edges)$ be a finite hypergraph, $t$ a positive integer and $\epsilon\in(t /{\card V},1)$. A family $S \subseteq \binom V t$ of $t$-subsets of $V$ is an \emph{$\epsilon$-$t$-net} for $H$ if for every $e \in \edges$ with $\card e \geq \epsilon \card V$ there is an $s\in S$ such that $s\subseteq e$.
\end{definition}

As mentioned already, for $t=1$ this is equivalent to the $\epsilon$-net notion, and for $t=\Theta(\epsilon \card V)$ this corresponds to the notion of $\epsilon$-Mnets. 
In this paper we study the following problem.

\begin{problem}
How small are the smallest $\epsilon$-$t$-nets for $H$? Can we compute them efficiently?
\end{problem}

\subsubsection*{Motivation.}
Instances of the $\epsilon$-$t$-net problem appear naturally in various contexts in computer science and combinatorics. For example, the following is a basic motivating example for \emph{secret sharing} \cite{Liu68,Shamir79}: \enquote{Eleven scientists are working on a secret project. They wish to lock up the documents in a cabinet so that the cabinet can be opened if and only if six or more of the scientists are present. What is the smallest number of locks needed?}. Consider a variant of this question in which the number of scientists is large. We still insist on the basic security condition -- that no less than six scientists can open the cabinet. On the other hand, due to the large number of scientists, we do not require that any six should be able to do so, but rather any sufficiently large group of a certain kind, e.g., at least one tenth of all scientists including a representative of each university involved.

The classical secret sharing methods (see, e.g., \cite{Beimel11}) distribute \enquote{keys} to subsets of 6 scientists so that any six scientists will be able to open the cabinet but no five will be able to do that. But as we require only certain groups of scientists to be able to open it, it is possible to distribute shared keys to only some of the 6-subsets. The questions: \enquote{What is the minimal number of 6-subsets we can achieve? and how can we choose the 6-subsets of scientists we distribute keys to?} are an instance of the $\epsilon$-$t$-net problem -- with $t=6$, $\epsilon=1/10$, and the hyperedges of the hypergraph being all groups of scientists that are required to be able to open the cabinet.

Other contexts in which the $\epsilon$-$t$-net problem appears (described in \cref{sec:applications}) include the Tur\'{a}n numbers of hypergraphs, $\chi$-boundedness of graphs, edge-coloring of hypergraphs and more.

\subsubsection*{Related work: $\epsilon$-Nets and Mnets.}

For any $t$, the minimum size of an $\epsilon$-$t$-net is sandwiched between the corresponding minimum sizes of $\epsilon$-nets and of Mnets. Indeed, given an Mnet, one obtains an $\epsilon$-$t$-net by picking one $t$-subset from each subset, and given an $\epsilon$-$t$-net, one obtains an $\epsilon$-net by taking one vertex from each $t$-subset. The survey \cite{MV17} has most known bounds on these objects. 

\subsection{Results}

\textbf{Notation}: we write $O_{x,y}(\cdot)$ when the implicit constants depend on parameters $x$ and $y.$ 

\subsubsection*{Hypergraphs of finite VC-dimension have small $\epsilon$-$t$-nets.} 
Our main result is an existence result for small $\epsilon$-$t$-nets.

\begin{restatable}{theorem}{mainlarget}
\label{thm:main_ht>2}
For every $\epsilon\in(0,1)$ and $t \in \mathbb{N}\setminus\{0\}$, every hypergraph on $\geq C_1 \left(\frac{t-1}{\epsilon}\right)^{d^*}$ vertices with VC-dimension $d$ and dual shatter function $\pi^*_H(m)\leq C m^{d^*}$ admits an $\epsilon$-$t$-net of size $O(\frac{d (1+\log t)}{\epsilon} \log \frac{1}{\epsilon})$, all elements of which are pairwise disjoint. Here $C_1=C_1(d^*,C)$.
\end{restatable}

(The dual shatter function, described in \cref{sec:tuple}, is a property of the hypergraph such that we may always take $d^*< 2^{d+1}$.)

This bound is asymptotically tight when $t=O(1)$, in the sense that there exist hypergraphs for which any $\epsilon$-net, and consequently also any $\epsilon$-$t$-net, is of size $\Omega(\frac{1}{\epsilon} \log \frac{1}{\epsilon})$ \cite{KomPaW92}. The proof of Theorem~\ref{thm:main_ht>2} involves a surprising relation between the $\epsilon$-$t$-net problem and the existence of \emph{spanning trees with a low crossing number}, proved by Welzl in 1988 \cite{Welzl88}.

Hypergraphs with VC-dimension 1 admit $O(\frac{1}{\epsilon})$-sized $\epsilon$-nets \cite{KomPaW92} and $\epsilon$-Mnets \cite{DGJM19}. The latter fact yields the following result, albeit with worse constants. We offer a simple proof.

\begin{restatable}{theorem}{vcone}
\label{thm:vcone}
For every positive integer $t$ and $\epsilon\leq \frac 1 2$, every finite hypergraph on $\geq t \lceil\frac 1 \epsilon\rceil$ vertices with VC-dimension 1 admits an $\epsilon$-$t$-net of size at most $t\lceil \frac 1 \epsilon\rceil +1$.
\end{restatable}

\subsubsection*{An efficient explicit construction of $\epsilon$-$t$-nets.}

Our second result is a new explicit construction of $\epsilon$-$t$-nets, for all $t \geq 1$.
The case of $t=1$ (i.e., $\epsilon$-nets) is of independent interest, as in this case our construction does not follow the proof strategy of Haussler and Welzl and does not use derandomization (unlike all previously known explicit constructions of $\epsilon$-nets). On the other hand, it has a sub-optimal size of $O_d(\frac{1}{\epsilon^d})$, where $d$ is the VC-dimension of the underlying hypergraph.

For a higher $t$, we introduce a new parameter of the hypergraph, which we call the \emph{$t$-VC-dimension}. For hypergraphs of $t$-VC-dimension $d$, we construct $\epsilon$-$t$-nets of size $O_{d}(\frac{1}{\epsilon^{d +t -1}})$. We give some first results on the relation between this new parameter and the standard VC-dimension.

\subsubsection*{Small $\epsilon$-2-nets for geometric hypergraphs.} 

In view of \cref{thm:main_ht>2}, which shows that for hypergraphs with a constant VC dimension one can obtain an $\epsilon$-$t$-net of roughly the same size as the smallest $\epsilon$-net, it is natural to ask whether a similar result can be achieved for geometrically-defined hypergraphs that admit an $\epsilon$-net of size $O(\frac{1}{\epsilon})$. We obtain such results for several geometrically-defined hypergraphs in $\Re^2$, including the intersection hypergraph of two families of pseudo-disks and the dual hypergraph of a family of regions with linear union complexity. Namely, we show that these hypergraphs have $O(\frac{1}{\epsilon})$-sized $\epsilon$-2-nets provided they have $\Omega(\frac{1}{\epsilon})$ vertices.
Interestingly, in some scenarios the minimum size of an $\epsilon$-2-net is sensitive to the exact multiplicative constant: there are subhypergraphs on $\Theta(\frac{1}{\epsilon})$ vertices for which any $\epsilon$-2-net is of size $\Omega(\frac{1}{\epsilon^2})$. 

\section{Construction of Auxiliary Hypergraphs}
\label{sec:tuple}

\subsection{Some preparatory results}

\subsubsection*{Sauer's lemma.}
Given a hypergraph $H=(V,\edges)$ the \emph{trace (also known as projection or restriction) of $H$ on $A\subseteq V$} is  $\Pi_H(A)=\{A \cap e \colon e \in \edges\}$; shattered subsets are those for which $\Pi_H(A)=2^A$.
The shatter function of $H$ is
\[\pi_H\colon n\in \Nat\mapsto \max\{  \card{\Pi_H(A)} : A\subseteq V,\ \card{A}\leq n \}.\]
It is bounded by the \emph{Sauer--Shelah lemma}:

\begin{lemma}[\cite{VC71,Sau72,She72}]\label{lm:sauer}
  If $\VC H = d$ then $\pi_H(n) \leq \binom{n}{0}+\binom{n}{1}+ \dots +\binom{n}{d}$.
  In particular, for $ 1\leq d \leq n$ one has $\pi_H(n) \leq (\frac{e} d)^d \cdot n^d$, where $e$ is Euler's number. 
\end{lemma}

\subsubsection*{Binary entropy function.} This is $h\colon x \in (0, 1) \mapsto -x \log x -(1-x) \log (1-x)$. (All logarithms are binary. See \cref{fig:entropy}.) We will use the following inequality.
\begin{equation}
    \label{cl:vc_pairs}
\forall \alpha \in \left(0, \frac{1}{2}\right],\ \forall n\in\Nat, \quad \log{\sum_{i=0}^{\lfloor \alpha n\rfloor} \binom{n}{i}} \leq {n h(\alpha)}.
\end{equation}

\begin{proof}
For the sake of simplicity, we assume $\alpha n \in \mathbb{N}$. By the binomial theorem,
\begin{align*}
  1=(\alpha+(1-\alpha))^n 
  & = \sum_{i=0}^{n} \binom{n}{i}  \alpha^i  (1-\alpha)^{n-i}\\
  & \geq \sum_{i=0}^{\alpha n} \binom{n}{i}  \alpha^i  (1-\alpha)^{n-i}=
  \sum_{i=0}^{\alpha n} \binom{n}{i}  (1-\alpha)^{n}    \left(\frac{\alpha}{1-\alpha}\right)^i \\
  &\geq \sum_{i=0}^{\alpha n} \binom{n}{i}  (1-\alpha)^{n}    \left(\frac{\alpha}{1-\alpha}\right)^{\alpha n} \quad \text{since $0<\frac{\alpha}{1-\alpha}<1$} \\ 
  &= 2^{-n \cdot h(\alpha)} \cdot \sum_{i=0}^{\alpha n} \binom{n}{i}.\qedhere
\end{align*}
\end{proof}

The binary entropy function restricted to $(0,\frac 1 2]$ is invertible, and \cite[Th.\@ 2.2]{Calabro}:
\begin{equation}
\forall x\in(0, 1), \quad \frac{x}{2 \log \frac 6 x} \leq h^{-1}(x) \leq \frac{x}{\log \frac 1 x}.
\end{equation} 

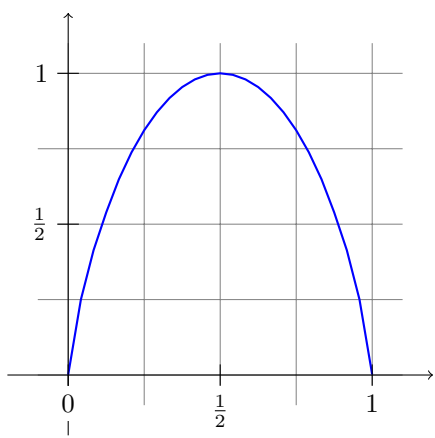
\begin{figure}
    \centering
\begin{tikzpicture}[domain=0.0001:0.9999,scale=4]
\draw[very thin,color=gray] (-0.1,-0.1) grid[step=0.25] (1.1,1.1);
\draw[->] (-0.2,0) -- (1.2,0); \draw[->] (0,-0.2) -- (0,1.2);
\draw[color=blue,thick]   plot (\x,{-\x*log2(\x) - (1-\x)*log2(1-\x)});
\foreach \x/\xtext in {0, 0.5/\frac{1}{2}, 1} 
  \draw (\x,1pt) -- (\x,-1pt) node[anchor=north,fill=white] {$\xtext$};
  \foreach \x/\xtext in {0.5/\frac{1}{2}, 1} 
  \draw (1pt,\x) -- (-1pt,\x) node[anchor=east,fill=white] {$\xtext$};
\end{tikzpicture}
    \caption{The binary entropy function.}
    \label{fig:entropy}
\end{figure}

\subsection{A first hypergraph on \texorpdfstring{$t$-subsets}{t-subsets}}
\begin{definition}
Given a hypergraph $H=(V,\edges)$ and a positive integer $t$, let $H^t$ be the hypergraph $(V^t,\edges^t)$ where $V^t=\binom{V}{t}$ and $\edges^t=\big\{\binom{e}{t}\colon e\in \edges\big\}$. That is, its vertices are all $t$-element subsets of $V$ and each hyperedge of $H^t$ consists of all such subsets contained in a given hyperedge of $H$. 
\end{definition}

For $t\in\Nat\setminus\{0,1\}$, let $\gamma_t = (t h^{-1}(1/t))^{-1}$. Note that $\log t \leq \gamma_t\leq 2 \log 6t$.

\begin{proposition}
\label{thm:vc_tuples}
If $H$ is a hypergraph with $\VC H=d$ then $d-t+1 \leq \VC H^t \leq \gamma_t d$.
\end{proposition} 

\begin{proof}
We assume that $t\geq2$, as for $t=1$, $\VC H^t = \VC H^1 =\VC H =d$. 

To prove the left inequality, let $\{v_1,\ldots,v_d\}$ be a shattered subset of vertices in $H$, with $d\geq t-1$. There are $d-t+1$ sets containing all vertices in $\{v_1,v_2,\dots,v_{t-1}\}$ and exactly one in $\{v_t,v_{t+1},\dots v_d\}$. It is easy to see that they form a shattered subset in $H^t$.

For the right inequality, suppose to the contrary that $P$ is a shattered set in $H^t$ with $d'=\card{P}>\gamma_t d$. Let $S=\cup_{p\in P}p$; clearly $\card{S} \leq td'$. Observe also that $d'+t-1 \leq \card S$. If this were not the case there would exist some $p_1\in P$ such that $p_1\subseteq \cup_{p\in P\setminus\{p_1\}} p$, which would contradict the fact that $P$ is shattered. 

We denote $\card{S}=\beta d'$; we have $1<\beta \leq t$.

Since $P$ is shattered in $H^t$, each $P_1\subseteq P$ is of the form $P\cap \binom{e} t= \{ p \in P \colon p \subseteq (S\cap e)\}$ for some $e\in \edges$. Thus $\card{\Pi_H(S)}\geq \card{2^P}=2^{d'}$.

On the other hand, $\VC H =d$, and so by \cref{lm:sauer}, $\card{\Pi_H(S)} \leq \binom{\beta d'}{0}+\binom{\beta d'}{1}+ \dots +\binom{\beta d'}{d}$.
It follows from \cref{cl:vc_pairs} (with $\beta d' \geq \beta \gamma_t d > 2d$)  that  $d'\leq\log \card{\Pi_H(S)} \leq   \beta d' h(\frac{d}{\beta d'})$. 
 We show that $1 > \beta h(\frac{d}{\beta d'})$, a contradiction. 

Note that $\frac{1}{t\gamma_t} \leq \frac{1}{\gamma_t \beta} < \frac 1 2$. Since $t\mapsto\frac{h(t)}{t}$ is monotone decreasing in the range $(0,1)$, we have $\gamma_t\beta \cdot h(\frac{1}{\gamma_t \beta})\leq t\gamma_t \cdot h(\frac 1 {t\gamma_t})=\gamma_t$. As $h$ is increasing on $(0,\frac 1 2)$, it follows that $\beta h(\frac{d}{ \beta d'})<\beta h(\frac 1 {\beta \gamma_t})\leq 1$. 
\end{proof}

\Cref{thm:vc_tuples} allows us to slightly improve the \enquote{trivial} upper bound of $O(\frac{d^t}{\epsilon^t}(\log \frac{1}{\epsilon})^t)$ on the minimum size of an $\epsilon$-$t$-net for any hypergraph with constant VC-dimension.
\begin{corollary}
\label{thm:UB}
Let $H$ be a hypergraph on $n$ vertices with VC-dimension $d$. For any $t,\epsilon$ such that $n \geq \frac{t}{\epsilon}$, $H$ admits an $\epsilon$-$t$-net of size $O(\frac{d t(1+\log t)}{\epsilon^t} \log \frac{1}{\epsilon})$.
\end{corollary}

Indeed, observe that an $\epsilon^t$-net for $H^t$ is an $\epsilon$-$t$-net for $H$, and apply the classical $\epsilon$-net theorem to $H^t$.

\subsection{A smaller, well-behaved hypergraph on \texorpdfstring{$t$}{t}-subsets}
A \emph{spanning cycle} $P$ for $H=(V,\edges)$ is a cycle graph on $V$ that visits all vertices (exactly once). For $e\in\edges$, let $\operatorname{cr}(P,e)$ be the number of edges of $P$ with one endpoint in $e$ and the other in $V\setminus e$. The \emph{crossing number} of $P$ with respect to $H$ is $\sup\{\operatorname{cr}(P,e) \colon e\in\edges\}$. 

The \emph{dual hypergraph} of $H$ is $H^*=(\edges,\edges^*)$, where $\edges^*$ consists of all hyperedges $v^*=\{e \in \edges \colon v \in e\}$ for $v\in V$. Its shatter function is the \emph{dual shatter function} of $H$, and is denoted by $\pi^*_H$.

If $\VC H =d$ then $\VC H^* \leq 2^{d+1}$ \cite{Assouad}, and hence $\pi^*_H(m) \leq C_d m^{2^{d+1}}$ for every positive $m$, where $C_d$ is a constant depending on $d$. In particular, any hypergraph with finite VC-dimension satisfies the hypotheses of the following theorem.

\begin{theorem}[{\cite[Lemma 3.3 and Theorem 4.2]{Welzl88}}]
\label{thm:lc_tuples}
Let $H$ be a hypergraph on $n$ vertices such that $\pi^*_H(m)\leq Cm^d$ for some constants $C>0$ and $d>1$. Then there exists another constant $C_1$ (depending on $C$ and $d$) and a spanning cycle for $H$ with crossing number $\leq C_1 n^{1-\frac{1}{d}}$.
\end{theorem}

(An additional $\log n$ factor in Welzl's original result was later removed \cite[Sec.\@ 5.4]{Ma99}. Up to constant factors, this theorem is equivalent to the same result for paths or trees.)

\begin{definition}
Let $H=(V,\edges)$ be a finite hypergraph with $\pi^*_H(m)\leq Cm^d$. Let $P$ be a spanning cycle for $H$ whose crossing number is minimal (and thus $\leq C_1 \card{V}^{1-\frac 1 d}$). Fix an arbitrary starting point $v_0\in P$ and orientation of $P$. For $0\leq i < \card V$, let $v_i\in V$ be the $i$-th vertex along $P$. Let $V^t_{lc}= \{ \{v_{kt},v_{kt+1},\dots, v_{kt + t -1}\} \colon 0 \leq k <\lfloor \frac {\card V} t\rfloor\}\subsetneq \binom V t$ (where the subscript $lc$ stands for low crossing). Observe that its elements are pairwise disjoint. Let $H^t_{lc}$ be the hypergraph on $V^t_{lc}$ whose hyperedges are of the form $\{  v \in V^t_{lc} \colon v \subseteq e \}$ for each $e\in \edges$.
\end{definition}

\begin{remark}
\label{rem:dim<=dim}
In order to make $H^t_{lc}$ uniquely defined, $P$ is chosen arbitrarily from all suitable spanning cycles. As $H^t_{lc}$ is a subhypergraph of $H^t$, $\VC H^t_{lc} \leq \VC H^t$, and thus we also have $ \VC H^t_{lc} \leq \gamma_t \VC H$.
\end{remark}

\section{Existence of Small \texorpdfstring{$\epsilon$-$t$-Nets}{\textepsilon-t-Nets}}

\label{sec:opt_eps_2_net}

\mainlarget*

\begin{proof} For $t=1$, this is simply the $\epsilon$-net theorem. For higher $t$, let $H=(V,\edges)$ be such a hypergraph and $n=\card V$. Consider the hypergraph $H^t_{lc}$ defined in \cref{sec:tuple}. It has $\lfloor \frac n t\rfloor$ vertices and VC-dimension $\leq \gamma_t d$ (by Remark~\ref{rem:dim<=dim}), and thus admits an $\frac{\epsilon}{2}$-net of size  $O(\frac{\gamma_t d}{\epsilon} \log \frac{1}{\epsilon})$. We claim that any such $\frac \epsilon 2$-net $N\subseteq \binom V t$ is also an  $\epsilon$-$t$-net for $H$.

Indeed, the crossing number of the associated spanning cycle is $O_{C,d^*}(n^{1-1/d^*})$. Every hyperedge $e$ of $H$ with $\card{e} \geq \epsilon n$ fully contains at least $\lfloor \frac{\epsilon n}t \rfloor - O_{C,d^*}(n^{1-1/d^*})$ elements of $V^t_{lc}$, which is $\geq \frac{\epsilon n}{2t}$ as soon as $n= \Omega_{C,d^*}(\frac t \epsilon n^{1 - 1/d^*})$, or equivalenty (noting also that $2(t-1) \geq t$ for $t\geq 2$) when $n=\Omega_{C,d^*}({\frac {t-1} \epsilon}^{d^*})$.  One of these $t$-subsets is in $N$.
\end{proof}

\begin{remark}
In general, some fast growth of $n=\card V$ as a function of $\frac 1 \epsilon$ is necessary. For example, given any $\epsilon$ such that $\frac t \epsilon \in \Nat$, the complete $t$-uniform hypergraph on $\frac t \epsilon$ vertices does not have any $\epsilon$-$t$-net with fewer than $\binom {t/\epsilon} t$ elements. Moreover, there exist geometrically-defined hypergraphs that do not admit $\epsilon$-$2$-nets of size $o(\frac{1}{\epsilon^2})$ (see Figure~\ref{fig:LB} and subsection \ref{app:rectangles}). On the other hand, in \cref{sec:geom} we show that certain classes of geometrically-defined hypergraphs have \enquote{small} $\epsilon$-$t$-nets even for \enquote{small} values of $n$. 
\end{remark}

\subsubsection*{Small $\epsilon$-nets, small $\epsilon$-$2$-nets.}
A natural question arising from \cref{thm:main_ht>2} is whether any hypergraph that admits small $\epsilon$-nets must also admit $\epsilon$-$t$-nets of approximately same size. In general, the answer is negative. Take for example a hypergraph whose smallest $\epsilon$-net is of size $\Omega(\frac{1}{\epsilon} \log \frac{1}{\epsilon}  )$ (see~\cite{KomPaW92},~\cite{PachTa13}), and augment it by adding a vertex that belongs to all hyperedges. Clearly, this second hypergraph has the same VC-dimension and a one-element $\epsilon$-net, but any $\epsilon$-$2$-net is of size $ \Omega(\frac{1}{\epsilon} \log \frac{1}{\epsilon}  )$.

However, this example is quite artificial. In \enquote{natural} scenarios (and for sufficiently large vertex sets) the smallest $\epsilon$-nets and $\epsilon$-$2$-nets might still have approximately same size. In \cref{sec:geom} we show that this is the case for some geometrically-defined hypergraphs.

Another scenario in which there exist both an $\epsilon$-net and an $\epsilon$-$2$-net of size $O(\frac{1}{\epsilon})$ is when the VC-dimension of the hypergraph is 1. In this case, the existence of an $\epsilon$-net of size $O(\frac{1}{\epsilon})$  was proved in \cite{KomPaW92}. The next theorem could be derived from results on Mnets \cite{DGJM19}, at the cost of poor multiplicative constants. Here we give a simpler proof for it.
\vcone*
\begin{proof}
Let $(V,\edges)$ be such a hypergraph and $n=\card V$. Without loss of generality, $\min\{\card e \colon e\in\edges\}\geq \epsilon n$. For $1\leq i \leq t$, there exists an $\epsilon$-net $N_i$ that hits each $e\in\edges$ at least $i$ times, and $\card{N_i}=i\lceil \frac 1 \epsilon\rceil$. To see this let $N_1$ be an $\epsilon$-net of size $\lceil\frac 1 \epsilon\rceil$ \cite{KomPaW92}. In the hypergraph induced on $V\setminus N_i$ the hyperedges hit only $i$ times by $N_i$ have cardinality $\geq \epsilon n -i$, while the number of vertices is $n - i\lceil \frac 1 \epsilon\rceil$, for a ratio $\frac{\epsilon n -i}{n - i\lceil \frac 1\epsilon\rceil}\geq \epsilon$. Take an $\epsilon$-net $N$ of size $\lceil \frac 1 \epsilon\rceil$ for this hypergraph and let $N_{i+1}=N_i \cup N$.
Finally, let the desired $\epsilon$-$t$-net consist of one $t$-subset from each element of $\Pi_H(N_t)$ with $\geq t$ vertices, of which there are at most $\card{N_t}+1$ by \cref{lm:sauer}.
\end{proof}

\section{Deterministic Construction of \texorpdfstring{$\epsilon$-$t$-Nets}{\textepsilon-t-Nets}}
\label{sec:eps-net}

Let $H=(V,\edges)$ be a finite hypergraph with VC-dimension $d$, and fix $\epsilon\in (0,1)$. In this section we provide an explicit polynomial-time construction of $\epsilon$-nets that immediately implies an explicit construction of $\epsilon$-$t$-nets. The size is far from optimal, but the construction is simpler than previous explicit constructions, as it does not rely on packing numbers nor on pseudo-random choices.

\subsection{Deterministic construction of \texorpdfstring{$\epsilon$-nets}{\textepsilon-nets}}
\label{subsec:eps_net}

We start with the following definition:

\begin{definition}
\label{def:stab}
Let $A,B$ be two subsets of $V$. We say that $A$ \emph{stabs} $B$ if for every hyperedge $S \in \edges$ with $B \subseteq S$ we have  $S \cap A \neq \emptyset$.
\end{definition}

Let $S \in \edges$ be a hyperedge, $\card S\geq {d+1}$, and let $X \in\binom S {d+1}$. Since the VC-dimension is $d$ the set $X$ is not shattered. Notice that $X= X \cap S\in \Pi_H(X)$. We can also assume that $\emptyset\in\Pi_H(X)$, for otherwise $X$ is a transversal for $H$ of size $d+1$.
Hence there exists at least one non-trivial, proper subset $A \subsetneq X$ such that $(X \setminus A) \notin \Pi_H(X)$.
Equivalently, there is a non-trivial partition of $X$ into $A$ and $X \setminus A$ such that $A$ stabs $X \setminus A$. We say that $X$ is of type $\card A\in \{1,\ldots,d\}$. Note that $X$ could have several types.
By the pigeonhole principle, there is a type $i$ and a subset $A\in \binom S i$ such that a fraction $d^{-1}\binom {\card S} i^{-1}$ of the elements of $\binom S {d+1}$ are stabbed by $A$, hence the following lemma holds:
\begin{lemma}\label{stabbing}
Let $S$ be a hyperedge containing $\geq d+1$ vertices of $V$. Then there exists an integer $i \in \{1,\ldots,d \}$ and a subset $A\in \binom S i$ that stabs $\binom {\card S} {d+1} d^{-1} \binom {\card S} i^{-1}$ subsets of cardinality $d+1-i$.
\end{lemma}

\subsubsection*{Constructing $\epsilon$-nets.}
Put $n := \card{V}$. We construct an $\epsilon$-net of size $O_d(\frac{1}{\epsilon^d})$ as follows. Start with $N = \emptyset$.
As long as there is a hyperedge $S \in \edges$ with $\card{S} \geq \epsilon n$ and $S \cap N = \emptyset$, \cref{stabbing} asserts that some $i$-subset from $S$ stabs $\Omega_d((\epsilon n)^{d+1-i})$ subsets of $S$ with cardinality $d+1-i$ for an appropriate $i \in \{1, \ldots,d\}$. Add all elements of this subset to $N$; we call this a type $i$ iteration.

The resulting set is an $\epsilon$-net by construction. It is left to show that $\card{N} = O_d(\frac{1}{\epsilon^d})$. As each step of the construction adds at most $d$ vertices to $N$ it is enough to bound the number of iterations $T$. By the pigeonhole principle, at least $\frac{T}{d}$ of the iterations have the same type, say $i$.
After a type-$i$ iteration $N$ stabs an additional $\Omega_d((\epsilon n)^{d+1-i})$ subsets of cardinality $d+1-i$ none of which were previously stabbed. Since there are ${\binom n  {d+1-i}}$ subsets of cardinality $d+1-i$ we have $\frac{T}{d} = O_d({\binom n {d+1-i}} (\epsilon n)^{-(d+1-i)}) = O_d(\frac{1}{\epsilon^d})$.

\subsubsection*{Complexity analysis}

We analyze the running time of the above algorithm. We assume that for the algorithm we have a data structure which is the incidence matrix of the hypergraph $H$. Without loss of generality, each hyperedge of $\edges$ may be replaced with a subset of cardinality $\lceil \epsilon n\rceil$. This can be done in time $O(\epsilon n^{d+1})$ due to the fact that $\card \edges=O(n^d)$. 

We consider each $X\in\binom{V}{d+1}$. Firstly we check if there is a hyperedge $S\in \edges$ which contains $X$, if not, we continue to the next subset. If yes, we consider each of the $2^{d+1}-2$ proper subsets of $X$. Let $A\subset X$ be such a subset. We check if $X\setminus A$ is stabbed by $A$. We can do it by going over all $O(n^d)$ hyperedges of $H$. Hence, in total this pre-processing step takes $O(n^{d+1}\cdot 2^{d+1}\cdot n^d)=O_d(n^{2d+1})$ running time. While determining the type of any $(d+1)$-subset of $X$ and scanning all the hyperedges of the hypergraph, we maintain for any $i$-subset $A \subset X$ $(1 \leq i \leq d)$, a list of all the $(d+1-i)$-subsets of $X$ that $A$ stabs and their number. 

Consider some iteration of the algorithm and let $S\in {\edges}$ be such that $\card S\geq\epsilon n$ and $S\cap N=\emptyset$ where $N$ is the collection of elements found until this iteration. We find a subset $A\subset S$ of size at most $d$ which stabs the most subsets of size $(d+1)-\card A$.

The running time of each iteration is $O(\card{S}^d\cdot n^d)=O(\epsilon^d n^{2d})$. Hence in total the running time of the algorithm after the pre-processing step is $O_d(\frac{\epsilon^d n^{2d}}{\epsilon^d})=O_d(n^{2d})$. Hence the total running of the algorithm described in the previous section is $O_d(n^{2d})$.

\subsubsection*{Immediate applications to \texorpdfstring{$\epsilon$-$t$-nets}{\textepsilon-t-nets}}
\label{sec:sub:immediate}

The construction of $\epsilon$-nets in \cref{subsec:eps_net} gives two straightforward constructions of $\epsilon$-$t$-nets.
\begin{enumerate}
    \item \emph{Trivial construction.} Use the above algorithm to explicitly construct $t$ disjoint  $\epsilon$-nets of size $O_d(1/\epsilon^{d})$, and take all $t$-subsets of elements in their union that contain one element from each net. The resulting $\epsilon$-$t$-net is of size $O_d(1/\epsilon^{td})$.
    \item \emph{Construction via $H^t_{lc}$.} Use the above algorithm to explicitly construct an $\frac{\epsilon}{2}$-net for the hypergraph $H^t_{lc}$, which is an $\epsilon$-$t$-net for $H$ (as was shown in the proof of Theorem~\ref{thm:main_ht>2}). The resulting $\epsilon$-$t$-net is of size $O_{d,t}(1/\epsilon^{\VC H^t_{lc}})$.
(The cycle with a low crossing number required for constructing the hypergraph $H^t_{lc}$ can be found in polynomial time \cite{Welzl88,Ma99}). 
\end{enumerate}

\subsection{Deterministic construction of \texorpdfstring{$\epsilon$-$t$-nets}{\textepsilon-t-nets}}
\label{subsec:construction}

We present a direct construction of $\epsilon$-$t$-nets without passing through $\epsilon$-nets. For the sake of convenience, we start by for presenting the method for $t=2$.

The following definition extends the classical notion of VC-dimension.

\begin{definition}
Let $t$ be a positive integer. Also let $H=(V,\edges)$ be a hypergraph, and $T',T$ such that $T' \subseteq T \subseteq V$. We say that $T'$ is \emph{$t$-realized} by $H$ (with respect to $T$) if  $T'\cup S \in\Pi_H(T)$ for some $S\subseteq T$ such that $\card S< t$. We say that $T$ is \emph{$t$-shattered} by $H$ if every $T' \subseteq T$ is $t$-realized by $H$ (with respect to $T$). The $t$-VC-dimension of $H$, denoted by $\VC_t H$, is the maximal size of a vertex set that is $t$-shattered by $H$.
\end{definition}

Note that the $1$-VC-dimension is the standard VC-dimension. Moreover, the $t$-VC-dimension is at most the $(t+1)$-VC-dimension for any positive integer $t$. 
We use the following adaptation of \cref{def:stab}:

\begin{definition}
Let $H=(V,\edges)$ be a hypergraph. Given two vertex sets $A,B \subseteq V$, we say that $A$ \emph{2-stabs} $B$ if each hyperedge of $\edges$ that contains $B$ also contains at least two vertices from $A$.
\end{definition}

\begin{theorem}\label{thm:direct}
For a hypergraph $H=(V,{\edges})$ with 2-VC-dimension $d$, one can construct explicitly an $\epsilon$-$2$-net of size $O_d(1/\epsilon^{d-1})$.
\end{theorem}

\begin{proof}
Let $S \in \edges$ be a hyperedge and let $X \in \binom S {d+1}$.
Since the 2-VC-dimension is $d$ the set $X$ is not 2-shattered. Notice that $X= X \cap S$ and so $X$ and all elements of $\binom{X}{d}$ are 2-realized by $H$ with respect to $X$. For our purpose, we can also assume that $\emptyset$ is 2-realized by $H$ (with respect to $X$), for otherwise $\binom X 2$ is a transversal for $H$ of size $\binom {d+1} 2$.
This means that there is a partition, say $X = A \cup (X \setminus A)$, such that $A$ 2-stabs $X \setminus A$.
Let $i= \card{A}$. Note that $i \in \{2,\ldots,d\}$. We say that $X = A \cup (X\setminus A)$ is a type $i$ partition.
We need the following lemma, whose proof is similar to that of \cref{stabbing}.
\begin{lemma}\label{2stabbing}
Let $S$ be a hyperedge containing $\geq d+1$ vertices of $V$. Then there exists an integer $i \in \{2,\ldots,d \}$ and a subset $A \subset S$ with cardinality $i$ that 2-stabs $\frac{{\binom {\card S} {d+1}}}{(d-1) {\binom {\card S} i}}$ subsets $B$ of cardinality $d+1-i$.
\end{lemma}

\subsubsection*{Constructing $\epsilon$-$2$-nets} Let $H=(V,\edges)$ be as above and let $\epsilon > 0$ be fixed. Put $n = \card{V}$. We construct an $\epsilon$-$2$-net of size $O_d(\frac{1}{\epsilon^{d-1}})$ as follows. We start with a set $N = \emptyset$.
As long as there is a hyperedge $S \in \edges$ with $\card{S} \geq \epsilon n$ that does not contain any pair $\{v,w\} \in N$, for an appropriate $i \in \{2, \ldots,d\}$ we take an $i$-subset $A \subset S$ 2-stabbing $\Omega_d((\epsilon n)^{d+1-i})$ subsets of $S$ with cardinality $d+1-i$, and add to $N$ all $\binom i 2$ elements of $A$. We call this a type $i$ iteration. This is possible by \cref{2stabbing}.

The resulting set is an $\epsilon$-$2$-net by construction.
It is left to show that $\card{N} = O_d(\frac{1}{\epsilon^{d-1}})$.
In each step of the construction we add at most ${\binom d 2}$ pairs to $N$ so it is enough to bound the number of iterations $T$. By the pigeonhole principle, at least $\frac{T}{d-1}$ of the iterations have the same type, say $i$.
There are ${\binom n {d+1-i}}$ subsets of cardinality $d+1-i$, and in each of the at least $\frac{T}{d-1}$ type $i$ iterations we $2$-stab at least $\Omega_d((\epsilon n)^{d+1-i})$ additional subsets of cardinality $d+1-i$, so we have $\frac{T}{d-1} = O_d(\frac{{\binom n {d+1-i}}}{(\epsilon n)^{d+1-i}}) = O_d(\frac{1}{\epsilon^{d+1-i}})$ so $t = O_d(\frac{1}{\epsilon^{d-1}})$ (since $i \geq 2$). This completes the proof of \cref{thm:direct}.
\end{proof}

\subsubsection*{Complexity analysis.}
The only significant difference between the constructions of \cref{subsec:eps_net} and of \cref{subsec:construction} is the factor that depends on the size of the resulting net. Hence, the complexity of the algorithm in this section is bounded by $O_d(n^{2d})$, where $d$ is the $2$-VC-dimension of $H$.

\subsubsection{Extension of the Direct Construction of \texorpdfstring{$\epsilon$-$2$-Nets}{\textepsilon-2-Nets} to \texorpdfstring{$\epsilon$-$t$-Nets}{\textepsilon-t-Nets}}
\label{app:Extension-to-t}

Now we show how our deterministic construction of $\epsilon$-2-nets can be extended to $\epsilon$-$t$-nets for other values of $t$. The following argument is a direct adaptation of the argument from \cref{subsec:construction}.

\begin{definition}
Let $H=(V,\edges)$ be a hypergraph. Given two disjoint vertex sets $A,B \subset V$, we say that $A$ \emph{$t$-stabs} $B$ if each hyperedge $e \in \edges$ that contains $B$ must contain at least $t$ vertices from $A$.
\end{definition}

\begin{proposition}\label{prop:t-direct}
Let $H=(V,\edges)$ be a hypergraph with $t$-VC-dimension $d$. Then one can construct explicitly an $\epsilon$-$t$-net for $H$ of size $O_{d,t}(\frac{1}{\epsilon^{d+1-t}})$.
\end{proposition}

\begin{proof}
Let $S \in \edges$ be a hyperedge and let $X \subseteq S$ be a subset of $S$ with cardinality $d+1$.
Since the $t$-VC-dimension is $d$ then the set $X$ cannot be $t$-shattered. Notice that $X= X \cap S$ and so each $X' \subset X$ with $\card{X'}\geq d+2-t$ is $t$-realized by $H$. For our purpose, we can also assume that $\emptyset$ is $t$-realized by $H$ (with respect to $X$), for otherwise the set of all $t$-subsets in $X$ is a transversal for $H$ of size $\binom{d+1} t$.
This means that there exists a subset $X \setminus A \subset X$ of size between $1$ and $d+1-t$ that is not $t$-realized by $H$ (with respect to $X$).
Equivalently, there is a partition, say $X = A \cup (X \setminus A)$ such that $A$ $t$-stabs $X \setminus A$.
Let $i= \card{A}$. Note that $i \in \{t,\ldots,d\}$. We say that $X = A \cup (X\setminus A)$ is a type $i$ partition.
Note that there could be more than one type partition for the same set $X$. 
We need the following lemma:
\begin{lemma}\label{2stabbingt}
Let $S$ be a hyperedge containing $k \geq d+1$ vertices of $V$. Then there exists an integer $i \in \{t,\ldots,d \}$ and a subset $A \subset S$ with cardinality $i$ that $t$-stabs $\frac{{\binom k {d+1}}}{(d-t + 1) {\binom k i}} = \Omega_{d,t}(k^{d+1-i})$ subsets $B$ of cardinality $d+1-i$.
\end{lemma}

\begin{proof}
For each subset in $\binom S {d+1}$ there exists a partition with one of the above stabbing types. By the pigeonhole principle at least $\frac{{\binom k {d+1}}}{d-t+1} = \Omega(k^{d+1})$ of these subsets have the same type, say $i$.
Each such subset $X$ of type $i$ is charged by a $t$-piercing subset $A \subset X$ of cardinality $i$. Then by the pigeonhole principle there is a subset $A \subset S$ of cardinality $i$ that is charged at least $\Omega(\frac{k^{d+1}}{{\binom k i}}) = \Omega(k^{d+1-i})$ times. This means that $A$ 2-stabs $\Omega(k^{d+1-i})$ subsets $B \subset S$ of cardinality $d+1-i$, as asserted.
\end{proof}

\subsubsection*{Constructing $\epsilon$-$t$-nets} Let $H=(V,\edges)$ be as above and let $\epsilon > 0$ be fixed. Put $n = \card{V}$. We construct an $\epsilon$-$t$-net of size $O(\frac{1}{\epsilon^{d+1-t}})$ as follows. 
We start with a set $N = \emptyset$.
As long as there is a hyperedge $S \in \edges$ with $\card{S} \geq \epsilon n$ that does not include any $t$-subset of $N$, for an appropriate $i \in \{t, \ldots,d\}$ we take an $i$-subset $A \subset S$ $t$-stabbing $\Omega((\epsilon n)^{d+1-i})$ subsets of $S$ with cardinality $d+1-i$, and add to $N$ all $\binom i t$ elements of $A$. We call this a type $i$ iteration. This is possible by \cref{2stabbingt}.

Obviously the resulting set is an $\epsilon$-$t$-net by construction.
It is left to show that $\card{N} = O(\frac{1}{\epsilon^{d+1-t}})$.
In each step of the construction we add at most $\binom d t$ subsets to $N$ so it is enough to bound the number of iterations. Denote this number by $L$. By the pigeonhole principle, at least $\frac{L}{d-t+1}$ of the iterations have the same type, say $i$.
There are $\binom n {d+1-i}$ subsets of cardinality $d+1-i$ and in each of the at least $\frac{L}{d-t+1}$ type $i$ iterations we $t$-stab at least $\Omega((\epsilon n)^{d+1-i})$ additional subsets of cardinality $d+1-i$. We have that $\frac{L}{d-t+1} = O(\frac{{\binom n {d+1-i}}}{(\epsilon n)^{d+1-i}}) = O(\frac{1}{\epsilon^{d+1-i}})$ so $L = O(\frac{1}{\epsilon^{d+1-t}})$ (since $i \geq t$). This completes the proof.
\end{proof}

\subsection{\texorpdfstring{$t$}{t}-VC-dimension versus classical VC-dimension}
\label{subsec:VCvs2VC}

What can be said about the relation between VC-dimension and our newly introduced $t$-VC-dimension, for $t\geq 2$? By definition, 
$\VC H \leq \VC_2 H$.
As shown below ideas from Dudley's unpublished lecture notes \cite[Th.\@ 4.37]{Dudley99} yield $\VC_2 H\leq 2 \VC H +1$.
This is sharp for some small hypergraphs, such as that with vertex set $\{a,b,c\}$ and hyperedges $\{a\}$, $\{b,c\}$, $\{a,c\}$, and $\{a,b,c\}$, which has VC-dimension 1 but 2-VC-dimension 3.

\begin{claim}
Let $H(V,\edges)$ be a hypergraph then $\VC_2 H\leq 2 \VC H +1$.
\end{claim}
\begin{proof}
Assume that $V$ be $2$-shattered. We can show that for every $V'\in 2^V$ either $V'$ or $V\setminus V'$ is shattered. This yields the desired result by taking $V'$ of cardinality $\left\lceil \frac {\card V} 2\right\rceil$.

If $V'$ is not shattered, then there exists $S\in 2^{V'}\setminus\Pi_H(V')$. For any $T\in 2^{V\setminus V'}$, there is a set $W$, $\card W\leq 1$, such that $S\cup T\cup W\in\Pi_H(V)$, because $V$ is $2$-shattered. Since $S\notin \Pi_H(V')$ we must have $W\subseteq V'$. But then this implies $T\in \Pi_H(V\setminus V')$, that is, $V\setminus V'$ is shattered.
\end{proof}

 For general $t$, we conjecture that $\VC_t H \leq 2\VC H + 2t-1$. The reasoning below gives roughly $\VC_t H \leq 9.09 \max\{\VC H, t-1\}$.
 
Let $H$ be a hypergraph of finite VC-dimension with a largest $t$-shattered subset of vertices $T$. As $T$ is $t$-shattered, we have $2^T = \{e\setminus S\colon e\in\Pi_H(T),\ S\subseteq T,\ \card S < t\}$. This yields

\begin{equation*}
2^{\VC_t H}\leq \card{\Pi_H(T)} \cdot \sum_{i=0}^{t-1} \binom {\VC_t H} i 
\leq  \sum_{i=0}^{\VC H} \binom {\VC_t H} i  \cdot \sum_{i=0}^{t-1} \binom{\VC_t H}{i},
\end{equation*}
with the last inequality following from \cref{lm:sauer}. When $\VC_t H \geq 2\max\{t-1,\VC H\}$, applying \cref{cl:vc_pairs} gives
\begin{equation*} 1 \leq  h\left(\frac {\VC H} {\VC_t H}\right)+ h\left(\frac{t-1}{\VC_t H}\right).\end{equation*} 

From this inequality we obtain:
\begin{proposition}\label{prop:2vs1}
For $t\in\Nat\setminus\{0\}$, the $t$-VC-dimension of a hypergraph of VC-dimension $d$ is at least $d$, at most $2 \gamma_2 \max\{d,t-1\}$ (where $\gamma_2\simeq 4.54)$, and, as $d\to\infty$, at most $2d+o(d)$.
\end{proposition}

An interesting geometric example is the hypergraph $H$ whose vertex set is a finite subset of $\Re^{d-1}$ and whose hyperedges are induced by half-spaces. It is well-known that $\VC H=d$. 

More generally, we have $\VC_t H \leq td$ for all $t$. Indeed, by Tverberg's theorem (see, e.g., \cite{MATOUSEK}), every set $T$ of $td+1$ points in $\Re^{d-1}$ admits a partition into $t+1$ pairwise disjoint and non-empty sets $T = X \cup Y_1 \cup \dots \cup Y_t$ such that the intersection of their convex hulls is non-empty. No half-space can $t$-realize $X$ since any half-space that contains $X$ must contain at least one point from each $Y_i$, that is, at least $t$ points of $T\setminus X$. 

Therefore, for this hypergraph and $t=2$, the direct construction yields an $\epsilon$-$2$-net of size $O_d(1/\epsilon^{2d-1})$, while the trivial construction (described at the end of Section~\ref{subsec:eps_net}) yields only a weaker upper bound of $O_d(1/\epsilon^{2d})$. With good bounds on $\VC H^2_{lc}$, the construction via $H^2_{lc}$ (see again Section~\ref{subsec:eps_net}) might provide even smaller $\epsilon$-$2$-nets. In the plane (namely, where $d=3$), it follows from \cite{GITS19} that $\VC H^2_{lc} \leq 5$, and so the upper bounds obtained using the direct construction and using $H^2_{lc}$ are the same -- $O(1/\epsilon^{5})$. 

\section{Geometric \texorpdfstring{$\epsilon$-2-Nets}{\textepsilon-2-Nets}}
\label{sec:geom}

For a fixed $\epsilon>0$, any hypergraph with VC-dimension $d$ and $n \geq \frac{C_d}{\epsilon^{2^{d+1}}}$ vertices admits, by Theorem~\ref{thm:main_ht>2}, an $\epsilon$-$2$-net of size $O(\frac{d}{\epsilon} \log \frac{1}{\epsilon})$. This leaves open two interesting questions:

\begin{enumerate}
    \item In cases where the hypergraph admits an $\epsilon$-net of small size, say $O(\frac{1}{\epsilon})$, does it also admit an $O(\frac{1}{\epsilon})$-sized $\epsilon$-2-net (or, more generally, $\epsilon$-$t$-nets)?
    \item Does this extend to smaller values of $n$?
\end{enumerate}

In this section we answer both in the affirmative for several classes of geometrically-defined hypergraphs.

\begin{definition}\label{def:intersection_hypergraph}
Given two families $B$ and $R$ of sets, the intersection hypergraph $H(B,R)$ is the hypergraph on vertex set $B$, where any $r \in R$ defines a hyperedge $\{ b \in B  :  b \cap r \neq \emptyset \}$.
\end{definition}

Note that $H(B,R)$ and $H(R,B)$ are (in general) not isomorphic but dual to each other. Intersection hypergraphs are ubiquitous in discrete and computational geometry. Particular attention is given to the case where either $B$ or $R$ is a set of points, with $H(B,R)$ respectively known as a \emph{primal hypergraph} defined by $R$ or a \emph{dual hypergraph} defined by $B$. See the survey \cite{MV17} and the references therein.

We present below several intersection hypergraphs that admit $O(\frac{1}{\epsilon})$-sized $\epsilon$-nets, and prove that each of them has $\epsilon$-2-nets of the same size. Furthermore, while \cref{thm:main_ht>2} applies only to hypergraphs with a very large number of vertices, the geometric hypergraphs discussed do not have to contain \enquote{many} vertices in order to guarantee the existence of \enquote{small} $\epsilon$-2-nets. In some cases (see, e.g., subsection \ref{app:frames}), the behavior is sharp: we can point out two constants $c_1<c_2$ s.t. if the number of vertices satisfies $\card{V} \geq \frac{c_2}{\epsilon}$ the hypergraph admits an $O(\frac{1}{\epsilon})$-sized $\epsilon$-2-net, while for $\card{V}\leq \frac{c_1}{\epsilon}$, 
there exist hypergraphs from the same family that admit only $\epsilon$-2-nets of size $\Omega(\frac{1}{\epsilon^2})$.

\subsection{Non-piercing regions}
For our first example we consider a large class of geometric objects introduced by Raman and Ray \cite{RR18}. A \emph{family of non-piercing regions} is a family of regions of $\Re^2$ such that for any two regions $\gamma_1$ and $\gamma_2$ the difference $\gamma_1 \setminus \gamma_2$ is connected. (Each region may contain holes. See \cite{RR18} for the exact definitions.)

This extends the more familiar notion of pseudo-disks.

\begin{theorem}\label{thm:linear_nets_non_piercing}
The intersection hypergraph of two families $B$ and $R$ of non-piercing regions with $B$ finite admits an $\epsilon$-net of size $O(\frac{1}{\epsilon})$ and, if $\epsilon \card B\geq 2$, an $\epsilon$-2-net of size $O(\frac{1}{\epsilon})$.
\end{theorem}

The proof relies on several intermediary results. The first one is about an analogue of the Delaunay graph for non-piercing regions \cite{RR18}. The important specific case where the regions are pseudo-disks had already been studied \cite{ADEP,KS17,Kes18}.

\begin{definition}
  \label{def:support}
A \emph{planar support} for the hypergraph $(V,\edges)$  is a planar graph $G$ on the same vertex set $V$ such that any hyperedge in $\edges$ induces a connected subgraph of $G$.
\end{definition}

\begin{theorem}[\cite{RR18}]
\label{thm:RR}
Given two families $B$ and $R$ of non-piercing regions, $B$ finite, their intersection hypergraph $H(B,R)$ admits a planar support.
\end{theorem}

The following corollary has already been noted for families of pseudo-discs \cite{ADEP}.

\begin{corollary}
Given two families $B$ and $R$ of non-piercing regions, $\VC H(B,R)\leq 4$.
\end{corollary}

\begin{proof}
   Let $B'\subseteq B$ be a shattered subset of vertices in $H(B,R)$. As the non-piercing property is clearly hereditary, the hypergraph $H(B',R)$ also admits a planar support. For every pair of vertices in $B'$ there exists a hyperedge of $H(B',R)$ that contains these two vertices and no other. Following \cref{def:support} these two vertices must share an edge in any planar support of $H(B',R)$. Thus said \emph{planar} support is a complete graph on $B'$, forcing $\card{B'}\leq 4$.
\end{proof}
\begin{proof}[Proof of \cref{thm:linear_nets_non_piercing}]
First we observe that $H(B,R)$ has $\epsilon$-nets of size $O(\frac 1 \epsilon)$. Since $H(B,R)$ is finite, we may assume that $R$ is finite as well. To paraphrase from Pyrga and Ray \cite[Theorem 4]{PyrgaRay}, the following properties suffice:
\begin{itemize}
    \item For any $0<\epsilon<1$ and any $B'\subseteq B$, $H(B',R)$ admits an $\epsilon$-net whose size depends only on $\epsilon$.
    \item There exist constants $\alpha > 0$, $\beta\geq0$ and $\tau >0$ s.t.\@ for any $R'\subseteq R$ there is a graph $G_{R'}=(R',E_{R'})$ with $\card {E_{R'}}\leq \beta \card{R'}$ so that for any element $b\in B$ we have $m_b\geq \alpha n_b -\tau$, where $n_b$ is the number of regions of $R'$ intersecting $b$ and $m_b$ is the number of edges in $E_{R'}$ whose both endpoints (which are regions of $R'$) intersect $b$.
\end{itemize}

The first condition is verified because $\VC H(B',R)\leq 4$ for every $B'$. For the second one, let $\alpha=\tau=1$ and $\beta=3$, and let $G_{R'}$ be a planar support of $H(R',B)$. (Note the use of duality!) The inequalities follow from its planarity and the connectedness of the subgraph \enquote{cut out} by each $b\in B$.

Finally, to obtain an $\epsilon$-2-net, let $K_1\subseteq B$ be an $\epsilon$-net for $H(B,R)$ of size $O(\frac 1 \epsilon)$. Let $R'$ consist of the regions of $R$, if any, that intersect $\geq \epsilon \card B$ regions of $B$ but only one of $K_1$, and let $K_2$ be an $\frac \epsilon 2$-net for $H(B\setminus K_1,R')$ also of size $O(\frac 1 \epsilon)$. Then the desired $\epsilon$-$2$-net consists of all edges in a planar support of $H(K_1\cup K_2, R)$.
\end{proof}

\subsection{Small union complexity}

Next, we prove the existence of a small $\epsilon$-$2$-net for the intersection hypergraph of regions in the plane with linear union complexity and points (i.e.\@ the dual hypergraph defined by the regions). 

The union complexity of a family of objects is the function $\kappa:\Nat \to \Nat$ that sends each $n\in\Nat$ to the number of faces of all dimensions in the boundary of the union of $\leq n$ objects, maximized over all subsets of $\leq n$ objects.  If $\kappa(n)=O(n)$, we say that the family has \emph{linear union complexity}. Families with linear union complexity include, e.g., families of pseudo-discs: the boundary of the union of $n\geq 3$ pseudo-discs consists of at most $6n-12$ arcs and as many vertices \cite{KLPS}. 

The $(\leq k)$-level complexity of the family is defined by counting all faces included in at most $k$ objects (not just on the boundary).
To make these definitions precise, one needs to define faces and their dimension; see the survey by Agarwal, Pach and Sharir \cite{APS-union}.

A specific case of the following result could also be derived from previous results on Mnets \cite{DGJM19}, if one adds the additional assumption that the regions have bounded \enquote{semi-algebraic description complexity}. (The proof of \cite{DGJM19} is involved and uses algebraic arguments).

\begin{theorem}
\label{thm:union_complexity}
Let $L$ be a finite family of regions in $\Re^2$ with linear union complexity
and let $P \subseteq \Re^2$ be a set of points. If $ \card L \geq \frac{2}{\epsilon}$ then $H(L,P)$ admits an $\epsilon$-$2$-net of size $O(\frac{1}{\epsilon})$.
\end{theorem}

\begin{proof}
Let $n:=\card L$. First, construct a set $K\subseteq L$ of size $O(\frac 1 \epsilon)$ such that every \enquote{heavy} point of $P$ is included in at least two elements of $K$, as in the proofs of \cref{thm:vcone} or \cref{thm:linear_nets_non_piercing}. This relies on the existence of $\epsilon$-nets of size $O(\frac{1}{\epsilon})$  for $H(L,P)$, a result of Aronov, Ezra and Sharir \cite{AES}.

Since linear union complexity is a hereditary property, $K$ as a subset of $L$ also has linear union complexity. By a standard argument using the Clarkson--Shor theorem \cite{ClarksonS89}, the $(\leq 2)$-level complexity of $K$ is linear as well. Hence, by Euler's formula, the number of hyperedges of size $2$ in $H(K,P)$ (whose order of magnitude is equal to the number of $(\leq 2)$-level faces in the arrangement of $K$) is at most $c\card K$ for some constant $c$. By the pigeonhole principle, some region $d\in K$ participates in at most $c$ such hyperedges (i.e., pairs of regions). We pick these at most $c$ pairs of regions to be elements of the $\epsilon$-$2$-net we construct, and repeat the process for $K\setminus\{d\}$. 

We continue in this fashion until all elements of $K$ are removed, and set the $\epsilon$-$2$-net $N$ to be the set of pairs we picked. Clearly, $\card N = O(\card K) = O(\frac{1}{\epsilon})$. To see that $N$ is indeed an $\epsilon$-$2$-net, let $p$ be a point that belongs to at least $\epsilon n$ regions of $L$. By construction, $p$ belongs to at least two regions of $K$. Consider the process in which the elements of $K$ are gradually removed, until none of them are left. As a single region is removed at every step, we can look at the step in which the number of remaining regions that contain $p$ is reduced from $2$ to $1$. Since at that step $p$ is included in exactly two regions of the arrangement,
the corresponding pair of regions is added to the $\epsilon$-$2$-net. Hence, $p$ is covered by both elements of a pair in the $\epsilon$-$2$-net, as asserted. This completes the proof.
\end{proof}

\begin{remark}
By essentially the same argument, the hypergraph $H(L,P)$ admits an $\epsilon$-$t$-net of size $O_t(\frac{1}{\epsilon})$ for any constant $t\leq\epsilon \card L$.
\end{remark}

We can extend \Cref{thm:union_complexity} to a family $L$ with union complexity $\kappa(n)=n \cdot f(n)$. In this case, the size of the $\epsilon$-$2$-net is $O(\frac{1}{\epsilon}\cdot{\log f(\frac{1}{\epsilon})   \cdot f(\frac 1 \epsilon \cdot \log f(\frac{1}{\epsilon})) })$. For example, if $\kappa(n)=n\log n$, then one obtains an $\epsilon$-$2$-net of size $O(\frac 1 \epsilon \cdot \log \frac{1}{\epsilon} \cdot \log \log \frac{1}{\epsilon})$. 

Indeed, by \cite{AES}, the hypergraph $H(L,P)$ admits an $\epsilon$-net of size $O(\frac 1 \epsilon \cdot \log f(\frac{1}{\epsilon}))$. Let $n'= \frac 1 \epsilon \cdot \log f(\frac{1}{\epsilon})$. By the Clarkson--Shor theorem \cite{ClarksonS89}, the $(\leq 2)$-level complexity is bounded by $O(n' \cdot f(n'))$, hence there exists a region that participates in at most $f(n')$ hyperedges of order 2. This means that the size of the obtained $\epsilon$-$2$-net is bounded by $O(n' \cdot f(n'))
$.

\subsection{More Geometric \texorpdfstring{$\epsilon$-2-Nets}{\textepsilon-2-Nets}}
\label{app:geom}

\subsubsection{Frames}
\label{app:frames}
The next class of intersection hypergraphs we consider is that of points with respect to frames, where a frame is the boundary of an axis-parallel rectangle. 

\begin{proposition}
\label{thm:frames}
Let $P$ be a finite set of points of $\Re^2$ and let $F$ be a family of frames. If $\card{P} \geq \frac{5}{\epsilon}$, then $H(P,F)$ admits an $\epsilon$-$2$-net of size $\leq \frac 8 \epsilon - 2$.
\end{proposition}

\begin{proof}
  Let $n=\card P$. If at least $\epsilon n$ points lie on a same frame, then one of its four sides contains at least $\lceil \frac \epsilon 4 n\rceil\geq 2$ of them. Thus it is sufficient to take an $\frac \epsilon 4$-2-net for $P$ with respect to the family of all axis-parallel segments. For any such axis-parallel segment $\ell$,  take all pairs consisting of the $i\lceil \frac \epsilon 4 n\rceil$-th  and $(i+1)\lceil \frac \epsilon 4 n\rceil$-th vertices on $\ell$. In total, at most $\frac 8 \epsilon - 2$ pairs also suffice to pierce all axis-parallel segments.
\end{proof}

The interesting behavior here is that the requirement that the vertex set is large cannot be omitted.
Consider for example the set $P$ of $n$ points depicted in \cref{fig:LB}, and let $\epsilon = \frac{2}{n}$. For any pair $\{  p_1,p_2   \}  \subset P$ such that $p_1$ is in the first quadrant and $p_2$ is in the third quadrant, there exists a frame $r$ such that $r \cap P = \{  p_1,p_2   \}$. Hence, any $\epsilon$-2-net for the intersection hypergraph of $P$ and a sufficiently rich family of frames is of size $\geq \frac{n^2} 4 =\frac{1}{\epsilon^2}$.

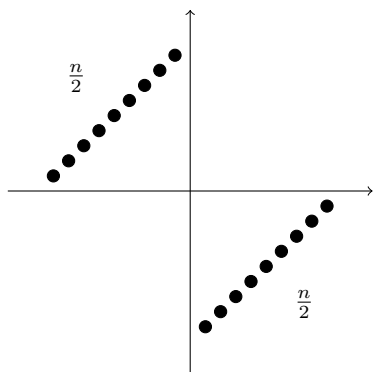
\begin{figure}
\centering
\begin{tikzpicture}[scale=2]
\draw[->] (-1.2,0) -- (1.2,0); \draw[->] (0,-1.2) -- (0,1.2);
\foreach \x in {0.1,0.2,...,0.9} {
\draw[fill=black] (1-\x,-\x) circle (0.04);
\draw[fill=black] (-\x,1-\x) circle (0.04);};
\node at (-0.75,0.75) {$\frac n 2$};
\node at (+0.75,-0.75) {$\frac n 2$};
\end{tikzpicture}

\caption{A set $P$ of $n$ points such that in the hypergraph of $P$ with respect to axis-parallel rectangles, for $\epsilon=\frac{2}{n}$, each $\epsilon$-2-net is of size $\Omega(n^2)=\Omega(\frac{1}{\epsilon^2})$.}
\label{fig:LB}
\end{figure}

\subsubsection{Axis-parallel rectangles}
\label{app:rectangles}

We conclude this section with the intersection hypergraph $H(P,R)$ of points and axis-parallel rectangles. This hypergraph admits an $\epsilon$-net of size  $O(\frac{1}{\epsilon}  \log \log \frac{1}{\epsilon})$ \cite{AES} and, if $\card{P} \gg \frac{1}{\epsilon^2}$, an $\epsilon$-2-net of size $O(\frac{1}{\epsilon} \log \frac{1}{\epsilon})$ by \cref{thm:main_ht>2}. (This last bound follows from the known fact that $\pi^*_{H(P,R)}(m)=\Theta(m^2)$, and from \cref{thm:main_ht>2}.) 

We extend the result to smaller values of $\card{P}$, at the expense of slightly increasing the size of the obtained $\epsilon$-2-net.

\begin{theorem}
\label{thm:rectangles}
Let $P$ be a finite set of points in $\Re^2$, and let $R$ be a family of axis-parallel rectangles. Assume that $\epsilon\card{P}\geq 3$. Then $H(P,R)$ admits an $\epsilon$-$2$-net of size $O(\frac{1}{\epsilon}  \log \frac{1}{\epsilon}  \log \log \frac{1}{\epsilon})$.
\end{theorem}

\begin{proof}

  Let $n:=\card P$. First, we construct a set $K\subseteq P$ that intersects every \enquote{heavy} rectangle of $R$ at least three times. Let $K_1$ be an $\epsilon$-net of size $O(\frac{1}{\epsilon}  \log \log \frac{1}{\epsilon})$  for $H(P,R)$, whose existence is a result of Aronov, Ezra and Sharir \cite{AES}. Let $R'$ consist of those rectangles in $R$, if any, that contain $\epsilon n$ points of $P$ but only either one or two points of $K_1$. They contain at least $\epsilon n-2 \geq \epsilon n /3$ points of $P\setminus K_1$ which is greater than $1$ by our assumption that $\epsilon n \geq3$. Then let $K_2$ be an $\epsilon/3$-net for $H(P\setminus K_1, R')$. It hits every rectangle of $R'$ at least once. Finally let $R''$ consist of the rectangles in $R'$ that contain only two points of $K_1\cup K_2$ and take an $\epsilon/3$-net $K_3$ for $H(P\setminus K_1\cup K_2,R'')$. We let $K=K_1\cup K_2\cup K_3$. It contains $k=O(\frac{1}{\epsilon}  \log \log \frac{1}{\epsilon})$ points, and each rectangle of $R$ containing $\epsilon n$ points of $P$ contains at least three points of $K$. Thus we restrict ourselves to finding a $\frac 3 k$-$2$-net for $H(K,R)$.

As coordinate-wise monotone transformations of the plane do not affect the intersections between $K$ and $R$, we may assume $K \subset \{0,1,\dots,k-1\}^2$ and that the width and height of every rectangle in $R$ are in $\{1,2,\dots,k\}$. Define the aspect ratio of a rectangle as the ratio of its height to its width, and for each integer $i$ with $\card i \leq\log k$ let $R_i$ be the family of all axis-parallel rectangles of aspect ratio $2^i$. Every such $R_i$ is a family of pseudo-discs: the boundaries of any two of its elements intersect at most twice. 
Take a $\frac 2 k$-$2$-net of size $O(k)$ for each $R_i$, altogether $O(k \log k)$ pairs.
Each heavy rectangle of $R$ is also the union of two rectangles with the same aspect ratio $2^i$ for some $i$, one of which must contain at least two points of $K$. (This idea is borrowed from Ackerman and Pinchasi \cite{AP13}.) 

Hence, the pairs form a $\frac 2 k$-$2$-net of size for each $H(K,R_i)$, a $\frac 3 k$-$2$-net for $H(K,R)$ and an $\epsilon$-$2$-net for $H(P,R)$.
\end{proof}

The case of axis-parallel rectangles illustrates a common phenomenon with $\epsilon$-2-nets: the bounds on the size of $\epsilon$-2-nets worsen as the number of rectangles decreases until, in the smallest case, hypergraphs on only $\frac 2 \epsilon$ vertices may require as many as $\Omega(\frac 1 {\epsilon^2})$ pairs.

\section{Applications of \texorpdfstring{$\epsilon$-$t$-Nets}{\textepsilon-t-Nets}}
\label{sec:applications}

We now describe several contexts in which $\epsilon$-$t$-nets appear naturally, and present possible applications of our results.

\subsection{The Tur\'{a}n problem for hypergraphs}

Tur\'{a}n's celebrated theorem in graph theory determines the largest possible number of edges in a graph that does not contain the complete graph $K_t$ (for a fixed integer $t$) as a subgraph. In 1941, Tur\'{a}n raised a similar question for hypergraphs: the maximum number of hyperedges that a $t$-uniform hypergraph on $n$ vertices can possess without containing the complete $t$-uniform hypergraph on $k$ vertices as a sub-hypergraph is the \emph{Tur\'an number} $T(n,k,t)$.

While Tur\'{a}n's theorem for graphs is sharp, the problem for hypergraphs remains notoriously difficult. For $k>t>2$, there are no known closed expressions for $n\mapsto T(n,k,t)$ (whereas $t=2$ corresponds to graphs). Determining $T(n,4,3)$ is considered one of the major open problems of hypergraph theory \cite{Keevash}.

Computing $T(n,k,t)$ amounts to finding the cardinality of a smallest $(k/n)$-$t$-net for the complete $k$-uniform hypergraph on $n$ vertices. Indeed, let $H=(V,\binom V k)$ be said complete hypergraph and let $N\subseteq \binom V t$. The set $N$ of $t$-subsets is a $(k/n)$-$t$-net for $H$ if and only if every set of $k$ vertices of $V$ contains a $t$-subset in $N$. Equivalently, there is no set of $k$ vertices $U\in \binom V k$ such that $\binom U t \subseteq \binom V t \setminus N$. Still equivalently, the complement $\binom V t\setminus N$ is the set of hyperedges of a $t$-uniform hypergraph on $V$ that does not contain the complete $t$-uniform hypergraph on $k$ vertices. We conclude:
\begin{proposition}
The size of a smallest $(k/n)$-$t$-net for the complete $k$-uniform hypergraph on $n$ vertices is $\binom n t - T(n,k,t)$.
\end{proposition}

\subsection{Edge coloring of hypergraphs}

Ackerman, Keszegh and P{\'{a}}lv{\"{o}}lgyi \cite{AKP18} introduced the problem of coloring $t$-subsets of vertices in a hypergraph in such a way that each hyperedge contains $t$-subsets of all colors. They focused on coloring $2$-subsets (i.e., edges) in geometric hypergraphs, and in particular on coloring the pairs that are themselves hyperedges of the hypergraph. They obtained constant bounds on the number of colors required for various classes of geometric hypergraphs. One of their main results is the following.
\begin{theorem}[{\cite[Theorem 4]{AKP18}}]
	For every dimension $d$ and integers $t \geq 2$, $k \geq 1$ and $h$, there exists an integer $m$ with the following property: given a set $H$ of $h$ half-spaces in $\Re^d$, the $t$-subsets of every finite set of points $S$ in $\Re^d$ can be colored with $k$ colors such that every half-space of $H$ that contains at least $m$ points from $S$ contains a $t$-subset of points of each of the $k$ colors.
\end{theorem}

Using $\epsilon$-2-nets, we obtain a result in the same spirit for all hypergraphs with bounded VC-dimension. 

\begin{proposition}\label{Prop:Edge-coloring-rainbow}
	Let $H$ be a hypergraph on $n$ vertices with VC-dimension $d$ and let $\epsilon\in(\frac 2 n, 1)$. Then the pairs of vertices of $H$ can be colored with $\Omega (\frac{n^2 \cdot \epsilon^{4}}{d} (\log \frac{1}{\epsilon})^{-1})$ colors such that each hyperedge of $H$ of size at least $\epsilon n$ contains a pair of each color.
\end{proposition}

\begin{proof}
Consider the corresponding pair hypergraph $H^2$. To each hyperedge of $H$ of size $\geq \epsilon n$ corresponds a hyperedge of $H^2$ of size $\Omega(\epsilon^2 n^2)$. Take an $\frac{\epsilon^2}4$-net for $H^2$ consisting of $O(\frac d {\epsilon^{2}} \log \frac{1}{\epsilon})$ pairs of vertices, color all its pairs with one color and remove them, and then repeat the procedure. 

We can continue until $\Theta(\epsilon^2 n^2)$ pairs have been colored, which is $\Omega(\frac{n^2 \cdot \epsilon^{4}}{d} (\log \frac{1}{\epsilon})^{-1})$ steps. All remaining pairs then receive any arbitrary color.
\end{proof}

\begin{remark}
When there is a hyperedge of size $O(\epsilon n)$, any such coloring has $O(\epsilon^2 n^2)$ colors. This is within a factor $O(\frac d {\epsilon^2} \log \frac 1 \epsilon)$ of our lower bound.
\end{remark}

To put this result in a perspective, note that if the pairs of vertices of $H$ are colored in $n^2 \epsilon^4 (\log(1/\epsilon))^{-1})$ colors randomly (i.e., for each pair, we pick a color uniformly, independent of other pairs), then for each hyperedge of size $\epsilon n$, the probability that it does not contain a pair of a given color is approximately
\[
(1-n^{-2} \epsilon^{-4} \log(1/\epsilon))^{\epsilon^2 n^2/2} \approx \exp(-\epsilon^{-2}\log(1/\epsilon)/2),
\]
which is bounded away from zero. The probability that a hyperedge contains edges of all colors is extremely low, and thus it is not clear that all hyperedges should contain pairs of all colors with positive probability. Hence, the construction of \cref{Prop:Edge-coloring-rainbow} is stronger than the result obtained by a random coloring.

\subsection{\texorpdfstring{$\chi$}{\textchi}-Boundedness of graphs}

The \emph{chromatic number} $\chi(G)$ of a graph $G$ is the minimum number of colors needed to color the vertices of $G$ such that any two neighboring vertices have different colors. The \emph{clique number} $\omega(G)$ is the size of the largest complete subgraph in $G$. Obviously, we always have $\chi(G) \geq \omega(G)$. A family $\mathcal{F}$ of graphs is \emph{$\chi$-bounded} if this inequality is \enquote{not far from being tight}, namely, if there exists a \emph{binding function} $f\colon \Nat\to\Nat$ such that $\chi(G) \leq f(\omega(G))$ for every $G \in \mathcal{F}$. On the notion of $\chi$-boundedness, see the recent survey \cite{SS18}. 

One of the first results on $\chi$-boundedness is a theorem of Wagon \cite{Wagon80} showing that, for any fixed $k$, the class of graphs that do not have $k$ pairwise disjoint edges as an induced subgraph is $\chi$-bounded. The theorem is proved by induction on $k$. We present the base case $k=2$, along with its proof, which is needed for understanding our application.
\begin{theorem}\label{thm:Wagon}
If a graph $G=(V,E)$ does not contain two disjoint edges as an induced subgraph, then $\chi(G) \leq \frac{\omega(G)(\omega(G)+1)}{2}$.
\end{theorem}

\begin{proof}
Let $A\subset V$ be a set of vertices such that the graph induced on $A$ is a clique of size $\omega(G)$. For each $a\in A$, let $S_a\subseteq V\setminus A$ be the set of vertices which are adjacent to all of the vertices of $A\setminus \{a\}$ and are not adjacent to $a$. Note that, for each $a\in A$, the set $S_a$ is independent, as otherwise there would be a clique in $G$ of size $\omega(G)+1$. 

For each $\{a,b\}\subseteq A$, let $S_{a,b}\subseteq V\setminus A$ be the set of vertices adjacent to neither $a$ nor $b$. Note that for each $\{a,b\}\subseteq A$, $S_{a,b}$ is also independent, as otherwise we would get an induced copy of two disjoint edges in $G$.

Further note that each $v\in V\setminus A$ is not adjacent to some vertex of $A$ as otherwise we would find a clique of size $\omega(G)+1$ in G, and therefore is in one of the sets described above.

We define a coloring $c$ of $G$ as follows. Each vertex in $A$ receives a unique color. For each $a\in A$, we color the set $S_a$ with color $c(a)$. We assign a new color to each of the sets $S_{a,b}$, $\{a,b\}\in A$. It is easy to check that it is a proper coloring. Moreover, we used $\card A + \binom {\card A} 2 = \frac{\omega(G)(\omega(G)+1)}{2}$ colors.
\end{proof}
The proof for general $k$ is an easy inductive argument, which essentially repeats the base step presented above. The binding function it yields is a polynomial of degree $2(k-1)$.

 We observe that the argument used in the proof of \cref{thm:Wagon} can be improved using an \emph{edge-hitting set} (i.e., a set of edges such that each hyperedge contains at least one of them) for an appropriately chosen hypergraph.

Consider the last step of the proof. The crucial observation it uses is that any $v \in V \setminus (A \cup S)$ belongs to some $S_{a,b}$. However, it seems that many of the vertices belong to many sets $S_{a,b}$ (formally, if $v$ has $\card{A}-\ell$ neighbors in $A$, then it belongs to ${\binom{\ell}{2}}$ sets $S_{a,b}$). Hence, it is plausible that we can \enquote{cover} all vertices $v \in V \setminus (A \cup S)$ by a smaller number of sets $S_{a,b}$, and hence, reduce the total chromatic number.

This is achieved by using an edge-hitting set. Let $H$ be the hypergraph whose vertex set is $A$, and whose hyperedge set is $\{e_v\}_{v \in V \setminus (A \cup S)}$, where $e_v = \{a \in A: (a,v) \not \in \edges\}$. That is, each $v \in V \setminus (A \cup S)$ induces an hyperedge that consists of all its non-neighbors in $A$. Let $T=\{(w_i,w'_i)\}_{i=1,\ldots,m}$ be an edge-hitting set for $H$. We claim that in the third step of Wagon's proof, instead of taking all pairs $\{a,b\} \subset A$, it is sufficient to take the $m$ pairs $\{w_i,w'_i\}$. Indeed, let $v \in V \setminus (A \cup S)$. By Wagon's argument, $v$ has at least two non-neighbors in $A$, and hence, the hyperedge $e_v$ is of size $\geq 2$. Hence, it contains some edge $(w_i,w'_i) \in T$. By the definition of $H$, this implies $v \in S_{w_i,w'_i}$. Thus, for each $v \in V \setminus (A \cup S)$ we have $v \in S_{w_i,w'_i}$ for some $(w_i,w'_i) \in T$, and so, it is sufficient to color the $m$ sets $S_{w_i,w'_i}$.

This reduces the bound on the chromatic number obtained by Wagon from $\omega+ \frac{\omega(\omega-1)}{2}$ to $\omega + m$, where $m$ is the smallest size of an \emph{edge-hitting set} for $H$.

The notion of $\epsilon$-2-nets can be useful in this context in two ways.

 First, $\epsilon$-2-nets can be used to construct a good edge-hitting set, using the strategy of constructing an approximate hitting set from an $\epsilon$-net, pioneered by Br\"{o}nnimann and Goodrich \cite{BG95} and followed-up in numerous works (e.g., \cite{CV07,ERS05}). One possible way to do this is to consider the hypergraph $H^2_{lc}$ that corresponds to $H$ (see \cref{sec:tuple}), find an $\frac{\epsilon}{2}$-net for it (which is an $\epsilon$-2-nets for $H$, as was shown in Section \ref{sec:opt_eps_2_net}), and use the algorithm of \cite{BG95} to leverage them into an approximate hitting set of $H^2_{lc}$, which is a small edge-hitting set for $H$. Another possible way is working directly with $\epsilon$-2-nets of $H$ (in cases where this approach is advantageous over working with $H^2_{lc}$) and transforming them into an approximate edge-hitting set of $H$, using a variant of the algorithm of \cite{BG95}.

Second, we can use an $\epsilon$-2-net to cover all \enquote{sufficiently large} hyperedges of $H$, which corresponds to coloring all vertices of $V \setminus (A \cup S)$ that have sufficiently many non-neighbors in $A$. Then, it remains to color vertices that have many neighbors in $A$, and one may hope that since $G$ does not contain a pair of disjoint edges, they can be colored in a relatively small number of colors.

\subsection{Secret sharing}

The relevance of $\epsilon$-$t$-nets to secret sharing was described in the introduction. We would like to add a few remarks:
\begin{itemize}
\item The classical objective of secret sharing is that, for some threshold $t$, the secret may be recovered by any coalition of $t$ members, and none of $t-1$ members. In our variant, we do not require that every large coalition be able to recover the secret, but only a certain upward closed subset of coalitions (e.g., those containing at least $\frac{n}{10}$ members, a representative of each university, and a senior researcher). This natural generalization is known as hierarchical secret sharing \cite{Simmons88,Tassa07}.
    
\item Departing from classical work on secret sharing, our problem has \emph{two thresholds} -- a \enquote{necessary threshold} of 6 members, required to obtain the secret, and a \enquote{sufficient threshold} which consists of a set of coalitions that must be able to obtain the secret. There is no restriction for the coalitions \enquote{between the thresholds}; they may or may not be able to obtain the secret. To the best of our knowledge, such a \enquote{two thresholds} scheme has not appeared in the secret sharing literature, although it seems interesting and worthy of study.
    
\item Another deviation from the usual setting of secret sharing is that we seek to minimize the number of divided key shares. This goal is natural in settings where the keys are physical, as in the \enquote{scientists problem} described in the introduction, or in secret sharing schemes in which the generation and storage of key shares have a cost.
\end{itemize} 

\section{Discussion and Open Problems}
\label{sec:open}
A hypergraph $H$ with finite VC-dimension $d$ has $\epsilon$-2-nets of size $O(\frac{d}{\epsilon} \log \frac{1}{\epsilon})$  when $n$ is very large as a function of $\frac 1\epsilon$.
This upper bound is the best possible in general, and as we saw in \cref{sec:opt_eps_2_net} may also be best possible even if $H$ admits smaller $\epsilon$-nets. 
However, we conjecture that in any \enquote{reasonable} setting, (including, e.g., all the geometric scenarios discussed in \cref{sec:geom}, and all hypergraphs with hereditarily small $\epsilon$-nets), the existence of an $\epsilon$-net of some order of magnitude, implies the existence of an $\epsilon$-2-net of roughly the same order of magnitude.

Furthermore, we are not aware of any hypergraph in which the dependence of $n$ in $\frac{1}{\epsilon}$ has to be as large as in the assumption of \cref{thm:main_ht>2}. It may be interesting to extend our results to smaller values of $n$ (as a function of $\frac{1}{\epsilon}$), and to understand whether (as in some of the geometric cases discussed above), there exists a sharp threshold (as a function of $\frac{1}{\epsilon}$) such that if $n$ is above this threshold, then the hypergraph admits an $\epsilon$-2-net of size $\Tilde{O}(\frac{1}{\epsilon})$, but if $n$ is below it, then any $\epsilon$-2-net for the hypergraph contains at least $\Omega(\frac{1}{\epsilon^2})$ pairs.

\bibliographystyle{plainurl}
\bibliography{references} 

\end{document}